\pdfoutput=1
\documentclass{article}
\usepackage{arxiv}
\usepackage[utf8]{inputenc}
\usepackage[english]{babel}
\usepackage[T1]{fontenc}  
\usepackage{extarrows}
\usepackage{blindtext}
\usepackage{hyperref}       
\usepackage{url}           
\usepackage{booktabs}     
\usepackage{amsfonts} 
\usepackage{mathrsfs}
\usepackage{amssymb}
\usepackage{amsthm}
\usepackage{amsmath}
\usepackage{amsfonts} 
\usepackage{nicefrac}       
\usepackage{microtype}      
\usepackage{graphicx}

\title{Attractor Flow Versus Hesse Flow in Wall-Crossing Structures}

\author{Qiang Wang \thanks{wangqiang@shanghaitech.edu.cn}\,\,\,\thanks{Institute of Mathematical Sciences,\, ShanghaiTech University,\, Shanghai, China}}

\begin{document}

\maketitle

\newtheorem{theorem}{Theorem}[section]
\newtheorem{corollary}{Corollary}[theorem]
\newtheorem{lemma}[theorem]{Lemma}
\newtheorem{remark}{Remark}[section]
\newtheorem{proposition}[theorem]{Proposition}
\newtheorem{definition}{Definition}[section]

\begin{abstract}
We recast the physics discussions in the paper of Dieter Van den Bleeken \cite{svan2012bps} within the context of wall-crossing structure à la Kontsevich and Soibelman \cite{kontsevich2014wall}. In particular, we compare the Hesse flow given in \cite{svan2012bps} and the attractor flow on the base of the complex integrable system, and show that both can be used in the formalism of wall-crossing structure. We also propose the notions of dual Hesse flow and dual attractor flow, and show that under the rotation of the $\mathbb{Z}$-affine structure, the Hesse flow can be transformed into the dual attractor flow, while the attractor flow into the dual Hesse flow. This suggests its possible use in Mirror Symmetry. 
\end{abstract}


\section{Introduction}
\textit{Wall-crossing structure} ("WCS" for short) is a formalism proposed by Kontsevich and Soibelman in \cite{kontsevich2014wall} to encode the Donaldson-Thomas invariants ("DT invariants" for short, or "BPS invariants" in physics) and study their jumps when certain walls are crossed. In this formalism, the invariants could in principle be computed inductively by considering certain flow lines on the base $\mathcal{B}$ of a complex integrable system. Moreover, the base is endowed with $\mathbb{Z}$-affine structures, and the flow lines correspond to the straight lines in the $\mathbb{Z}$-affine structures. These flows are called the \textit{(split) attractor flows}, which derived its name from the "attractor mechanism" in the study of super-symmetric multi-centered black holes in $N=2, d=4$ super-gravity (see for example \cite{denef2011split}\cite{denef2000supergravity}\cite{moore1998arithmetic} for more about this story). 

In \cite{svan2012bps}, the pertinence of the Hesse potential (used in writing the K\"ahler metric on $\mathcal{B}$ as the Hessian of the Hesse potential) in generalizing the attractor flow is discussed, and the notion of \textit{Hesse flow} is introduced. In this paper, we recast the discussion there in the framework of WCS, and perform explicit computation in the $\mathbb{Z}$-affine coordinates adapted to special coordinates. This enables us to see more clearly the relationship between Hesse flow and the attractor flow. We will show that they are two manifestations of the same object in different $\mathbb{Z}$-affine structures (dual to each other). 

In section \ref{sec:2}, we review the notions of $\mathbb{Z}$-affine structure and special geometry to settle up the framework for later discussion. In section \ref{sec:3}, the notion of the complex integrable system will be reviewed. We also review briefly the formalism of wall-crossing structure established in \cite{kontsevich2014wall}, and indicate how it can be constructed from the data coming from complex integrable system. Along the way, the concept of the (split) attractor flow will be discussed with emphasis on its role in producing an algorithm for computing DT-invariants within the WCS formalism. 

Finally, in section \ref{sec:4}, we will use the special coordinates and the their adapted $\mathbb{Z}$-affine coordinates to perform explicit computations that enable us to write the attractor flow as Hesse flow and vise-verse (section \ref{sec:4.2}). The notions of dual attractor flow and dual Hesse flow will emerge naturally (section \ref{sec:4.3}). Besides, we will see how the computations can be interpreted as a duality between Monge-Amp$\grave{e}$re structures on the base manifold (proposition \ref{pro:4.4}).  

\section{Special geometry and \texorpdfstring{$\mathbb{Z}$}--affine structure}\label{sec:2}
In this section, we collect some definitions and terminologies for later use. The main references are \cite{freed1999special}\cite{alekseevsky2002special}\cite{gross2003affine}.
\subsection{\texorpdfstring{$\mathbb{Z}$}--affine manifold and SYZ mirror symmetry}

\begin{definition}
 An n-dimensional topological manifold $\mathcal{B}$ is called an \textit{integral affine ($\mathbb{Z}$-affine) manifold} if it is endowed with an (integral) affine structure  that is given by a maximal atlas $\mathcal{A}=\{U_{i},\Phi_{i}\}$ of coordinate charts such that the transition maps $\Phi_{j}\circ\Phi_{i}^{-1}$ belong to the group of integral affine transformations of $\mathbb{R}^{n}$,i.e.,$Aff_{\mathbb{Z}}(\mathbb{R}^{n}):=GL(n;\mathbb{Z})\ltimes\mathbb{R}^{n}$.
 \end{definition}

$\mathbb{Z}$-affine structure arises (\cite{gross2011real}\cite{gross2003affine}\cite{leung2005mirror}) when investigating mirror symmetry through the Stominger-Yau-Zaslow (SYZ) approach (\cite{strominger1996mirror}). In this approach, the mirror $X^{\vee}$ of a Calabi -Yau manifold $X$ is constructed, roughly speaking (modulo details coming from "instanton"-corrections), as the dual special Lagrangian fibration over the $\mathbb{Z}$-affine base $\mathcal{B}$.

Here $\mathcal{B}$ can be endowed with two $\mathbb{Z}$-affine structures induced from the complex and symplectic structures on $X$ and $X^{\vee}$ respectively. The model is obtained by considering the tangent bundle $\mathcal{T}_{\mathcal{B}}$ and the cotangent bundle $\mathcal{T}^{*}_{\mathcal{B}}$ away from singular locus $\Delta$ (where the fibrations degenerate). Suppose that $\mathcal{B}$ is endowed with $\mathbb{Z}$-affine coordinates $\left\{y_{1},\cdots,y_{n}\right\}$, then $X$ and $X^{\vee}$, as Lagrangian torus fibrations over the same base $\mathcal{B}$, can be described respectively as
 \[f:X=\mathcal{T}_{\mathcal{B}}/\Lambda\longrightarrow \mathcal{B}\qquad f^{\vee}:X^{\vee}=\mathcal{T}^{*}_{\mathcal{B}}/\Lambda^{\vee}\longrightarrow \mathcal{B}\]
where $\Lambda$ is the lattice spanned by $\left\{\frac{\partial}{\partial y_{1}},\cdots,\frac{\partial}{\partial y_{n}}\right\}$, and $\Lambda^{\vee}$ is the one spanned by $\left\{dy_{1},\cdots,dy_{n}\right\}$.

A \textit{Hessian metric} on $\mathcal{B}^{0}:=\mathcal{B}\backslash\Delta$ is given by the Hessian of a multi-valued convex function $K$, which consists of a set of convex functions $K_{i}:U_{i}\rightarrow\mathbb{R}$ on open convex cover $\{U_{i}\}$ of $\mathcal{B}^{0}$ such that the difference $K_{i}-K_{j}$ are affine linear on $U_{i}\cap U_{j}$. The Hessian metric lifts to a K\"ahler metric $g$ on $X$ with the associated K\"ahler potential $\omega=2i\partial\bar{\partial}K$. 
Hitchin (\cite{hitchin1997moduli}) observed that in order for $g$ to be Ricci-flatness, $K$ must satisfies the \textit{real Monge-Amp$\grave{e}$re equation}:
\[det\left(\frac{\partial^{2}K}{\partial y_{i}\partial y_{j}}\right)=constant.\]

\begin{definition}\label{def:2.2}
A integral affine manifold $\mathcal{B}$ endowed with a Hesse potential $K$ that satisfies the above Monge-Amp$\grave{e}$re manifold is called the \textit{affine Monge-Amp$\grave{e}$re manifold}.
\end{definition}

\subsection{Special geometry on \texorpdfstring{$\mathcal{B}$}{lg}}

The following well-known proposition gives an equivalent characterization for $\mathbb{Z}$-affine manifold.

\begin{proposition}\label{pro:2.1}
 Given an $n$-dimensional $\mathbb{Z}$-affine manifold $\mathcal{B}$, there is an torsion-free flat connection $\nabla$ on the (real) tangent bundle $T_{\mathbb{R}}\mathcal{B}^{}$, together with a maximal rank $\nabla$-covariant lattice $T^{\mathbb{Z}}_{\mathbb{R}}\mathcal{B}\subset T_{\mathbb{R}}\mathcal{B}$. And conversely, given such a connection and the lattice, we can induce a $\mathbb{Z}$-affine structure on $\mathcal{B}$.
\end{proposition}

\begin{definition}
(c.f., \cite{freed1999special}\cite{alekseevsky2002special}) A complex manifold $(\mathcal{B},J)$ is called a \textit{special complex manifold} if $\mathcal{B}$ is endowed with a flat torsionfree connection $\nabla$ on its real tangent bundle $T_{\mathbb{R}}\mathcal{B}$ such that $d_{\nabla}J=0$, where $J$ is the complex structure, and $d_{\nabla}$ the covariant derivative associated to $\nabla$. By a \textit{special symplectic manifold}, we mean a special complex manifold $(\mathcal{B},J,\nabla)$ together with $\nabla$-parallel symplectic form $\omega$. Finally, a \textit{special K\"ahler manifold} is a special symplectic manifold $(\mathcal{B},J,\nabla,\omega)$ for which $\omega$ is $J$-invariant, i.e., of type $(1,1)$. Then the induced metric $g(\cdot,\cdot):=\omega(J\cdot,\cdot)$ is called the \textit{special K\"ahler metric} of the special K\"ahler manifold  $(\mathcal{B},J,\nabla,\omega)$.
 \end{definition}

 \begin{lemma}\label{lem:2.2}
 (\cite{freed1999special}) The torsionfree condition of $\nabla$ is equivalent to the following condition\[d_{\nabla}(id)=0\]
 where $id\in\Gamma(T^{*}\mathcal{B}\otimes T\mathcal{B})$ is the $T\mathcal{B}$-valued one form that takes any tangent vector $X$ to itself.
  \end{lemma}
\begin{proof}
 \[d_{\nabla}(id)(X,Y)=(-1)^{1+1}\nabla_{X}id(Y)+(-1)^{1+2}\nabla_{Y}id(X)+(-1)^{1+2}id([X,Y])\]\[=\nabla_{X}Y-\nabla_{Y}X-[X,Y]\]
\end{proof}
 
Given a flat local framing $\{\xi_{\alpha}\}$ of $\mathcal{B}$, denote by $\{\eta^{\beta}\}$ the dual coframe, then $id=\sum_{\alpha}\xi_{\alpha}\otimes\eta^{\alpha}$, and $d_{\nabla}(id)=0$ implies that $d_{\nabla}\eta^{\alpha}=0$. By Poicar\'e lemma, there exist local coordinates $u^{\alpha}$ such that $du^{\alpha}=\eta^{\alpha}$, and these local coordinates satisfy the flat condition $\nabla du^{i}=0$. These coordinates are therefore called \textit{flat coordinates}, or by its connection to $\mathbb{Z}$-affine structure (see proposition \ref{pro:2.1}), they are $\mathbb{Z}$-\textit{affine coordinates}. This motivates the following definition:
 
 \begin{definition}\label{def:2.4}
 Let $(\mathcal{B},J,\nabla,\omega)$ be a special K\"ahler manifold. An affine local coordinates $(x^{i},y_{j})$ is called real special coordinates if equation (1) holds. Holomorphic coordinates $\{z_{i}\}$ is called special if its real part is real affine, i.e. $\nabla Re(dz^{i})=0$. We say that the special coordinates $z_{i}$ and the real special coordinates $\{x^{i},y_{j}\}$ are adapted if $Re(z^{i})=x^{i}$. Further more, two set of special coordinates $\{z^{i}\}$ and $\{w_{j}\}$ are said to be conjugate if there exists real special coordinates $\{x^{i},y_{j}\}$ such that $Re(z^{i})=x^{i}$ and $Re(w_{j})=y_{j}$.
\end{definition}
 
It is proved in \cite{freed1999special} that if$\{z^{i}\}$ and $\{w_{j}\}$ are conjugate coordinates which are adapted to the special coordinates $\{x^{i},y_{j}\}$, then their exists a local holomorphic function $\mathfrak{F}$, called the \textbf{holomorphic prepotential}, such that 
\begin{equation}
w_{j}=\frac{\partial\mathfrak{F}}{\partial z^{j}}
\end{equation}

 Explicit example will be given  in section 3.2 when discussing the complex integrable system with central charge.

\section{Complex integrable system and wall-crossing structure}\label{sec:3}

The complex integrable system can be viewed as Lagrangian torus fibration over the base $\mathcal{B}$. The $\mathbb{Z}$-affine structure and its dual on $\mathcal{B}$, together with the associated Hesse potentials, can be used to endow $\mathcal{B}$ with the structure of \textit{Monge-Amp$\grave{e}$re manifold} in two different ways naturally dual to each other under \textit{Legendre transformation}. This had been indicated by Kontsevich and Soibelman in \cite{kontsevich2001homological} in studying the homological mirror symmetry as the duality of torus fibration. 

A \textit{complex integrable system} (see for example section 4 of \cite{kontsevich2014wall}) is defined to be a holomorphic surjective map
\begin{equation}
    \pi:(X,\omega^{2,0})\longrightarrow \mathcal{B}
\end{equation} with regular fibers being holomorphic Lagrangian submanifolds. Here $(X,\omega^{2,0})$ is a holomorphic symplectic manifold of complex dimension $2n$, and the base $\mathcal{B}$, is a complex manifold of half-dimension. Denote by $\mathcal{B}^{0}\subset\mathcal{B}$ the dense open subset over which the fibers of $\pi$ being regular, and by \[\Delta=\mathcal{B}^{sing}:=\mathcal{B}\backslash\mathcal{B}^0\] the \textbf{discriminant locus} over which $\pi$ becomes degenerate. Assume that the fibres of $\pi$ are compact, so they are actually holomorphic Lagrangian tori.

\begin{figure}[h]
    \centering
    \includegraphics[height=2.8cm, width=5.3cm]{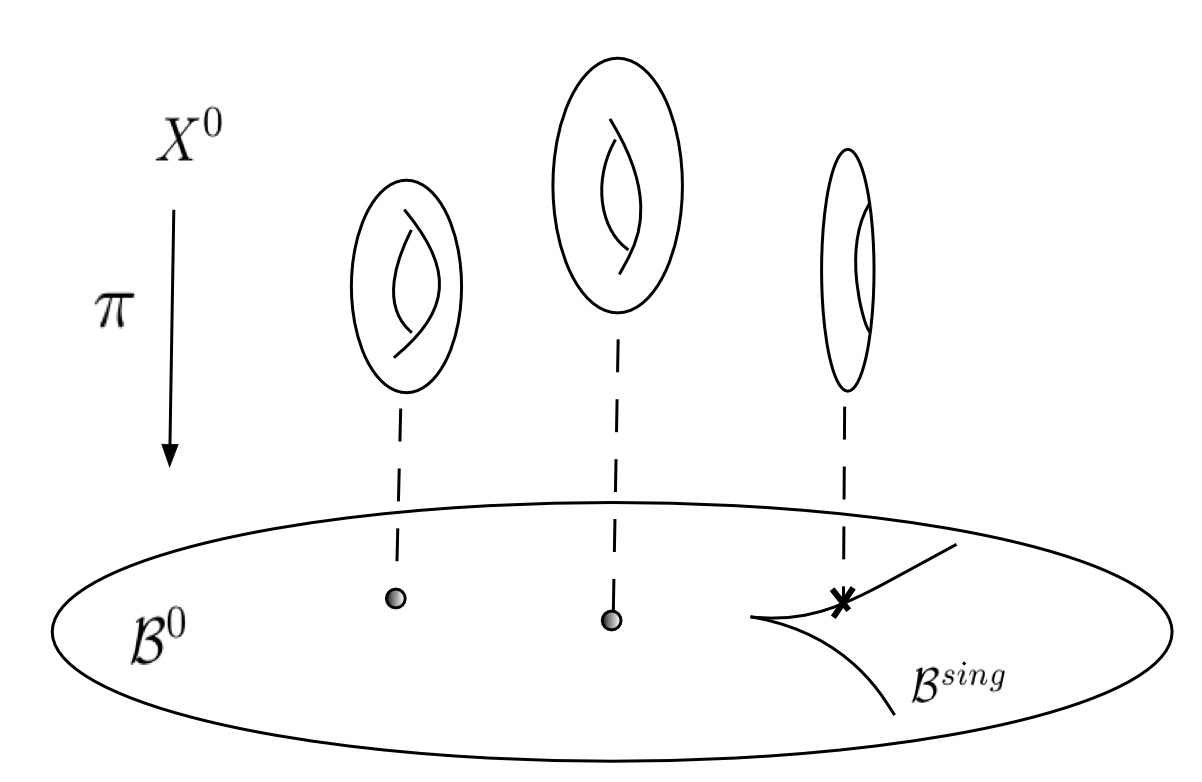}
    \caption{Complex integrable system}
\end{figure}

It is well known that the (smooth) base $\mathcal{B}^{0}$ of the complex integrable system $\pi$ is a $\mathbb{Z}$-affine manifold.

\subsection{Action-angle coordinates and central charges}

Viewing the integrable system as the real one by considering the symplectic form $\omega:=Re(\omega^{2,0})$. It is well-known that there exist action coordinates: $\{I^{1},\cdots,I^{n}\}$. Together with fiber coordinates $\{\theta_{1},\cdots,\theta_{2n}\}$ on the affine tori $\mathbb{T}^{2n}$, they form the \textit{action-angle coordinates}. In these coordinates, the symplectic form $\omega$ can be written in canonical form as 
  \begin{equation}
  \omega=\sum_{i}dI^{i}\wedge d\theta_{i}
  \end{equation}
We see that $\omega=d\alpha_{can}$, where $\alpha_{can}:=\sum_{i}I^{i}\,d\theta_{i}$ is the Liouville one form. And the fibrations is given by projection \[\pi:(I^{1},\cdots,I^{2n},\theta_{1},\cdots,\theta_{2n})\mapsto(I^{1},\cdots,I^{2n})\]

Denote by $\{\gamma_{i}\}$ the basis for $H_{1}(\pi^{-1}(b),\mathbb{Z})$ dual to $\{d\theta_{i}\}$, i.e., $\oint_{\gamma_{i}}d\theta_{j}=\,\delta_{ij}$. Then we have that
\begin{equation}
I^{i}=\oint_{\gamma_{i}}\alpha_{can}
\end{equation}

\begin{definition}
The local system of lattices $\underline{\Gamma}$ over $\mathcal{B}^{0}$, with stalk $\underline{\Gamma}_{b}$ being $H_{1}(\pi^{-1}(b),\mathbb{Z})$ is called the \textbf{charge lattice}.
\end{definition}

The action coordinates $\{I^{i}\}$ extends to an additive map $I$ from $\underline{\Gamma}$ to $\mathbb{R}$. We now show that the complex analogy of $I$ in will be the central charge associated to a complex integrable system. First, the following holomorphic forms are closed 
\begin{equation}\label{eq:5}
\alpha_{i}:=\oint_{\gamma_{i}}\omega^{2,0}
\end{equation}
So there exist local holomorphic functions $z_{i}$ defined up to an additive constant such that $dz_{i}=\alpha_{i}\,\, \,\text{for}\,\, 1\leq i\leq 2n.$
The set $\{z_{1},\cdots,z_{2n}\}$ are complex analogy of the action coordinates $\{I^{1},\cdots,I^{2n}\}$, and they define an holomorphic map in an open neighborhood $U$ of $b$:
\begin{equation}\label{eq:6}
(z_{1},\cdots,z_{2n}): U\longrightarrow\mathbb{C}^{2n}
\end{equation}

\begin{definition}
(\cite{kontsevich2014wall}) The collection of one forms $\alpha_{i}$ defined in (\ref{eq:5}) give rise to an element $\delta\in H^{1}(\mathcal{B}^{0},\underline{\Gamma}^{\vee}\otimes\mathbb{C})$, if $\delta=0$, there exists $Z\in\Gamma(\mathcal{B}^{0},\underline{\Gamma}\otimes\mathscr{O}_{\mathcal{B}^{0}})$, called the \textbf{central charge}, which satisfies $dZ(\gamma_{i})=\alpha_{i}$, for $1\leq i\leq2n$.
\end{definition}

\subsection{Special coordinates induced from the central charge}

Assuming that the fiber of our complex integrable system is endowed with a covariantly constant integer polarization. We call such integral system \textit{polarized complex integrable system}. A polarization gives a covariantly constant skew-symmetric bilinear form on $\underline{\Gamma}\otimes\mathbb{Q}$, i.e.,
\[\langle\cdot,\cdot\rangle: \Lambda^{2}\,\underline{\Gamma}\to\underline{\mathbb{Z}}\,_{\mathcal{B}^{0}}\]

Denote by $(\omega_{ij})$ the matrix $(\langle\gamma_{i},\gamma_{j}\rangle)$, and by $(\omega^{ij})$ its inverse. We have the following proposition given in \cite{kontsevich2014wall}:

\begin{proposition}
  a) $\sum_{i,j}\omega^{ij}dz_{i}\wedge dz_{j}=\sum_{i,j}\omega^{ij}\alpha_{i}\wedge\alpha_{j}=0.$
  \[b)\,\sqrt{-1}\sum_{i,j}\omega^{ij}dz_{i}\wedge d\bar{z}_{j}=\sqrt{-1}\sum_{i,j}\omega^{ij}\alpha_{i}\wedge\bar{\alpha}_{j} >0.\]
\end{proposition}

\begin{proposition}
  The condition a) in the above proposition implies that the image of the map (6) is a Lagrangian sub-manifold, while the condition b) implies that the map is an immersion. 
\end{proposition}
\begin{proof}
The symplectic form on $\underline{\Gamma}^{\vee}\otimes\mathbb{C}\cong\mathbb{C}^{2n}$ is given by $\Omega=\sum_{i,j}\omega^{ij}du^{i}\wedge du^{j}$, where $\{u^{1},\cdots,u^{2n}\}$ are the coordinates for $\mathbb{C}^{2n}$. a) implies that the image of (\ref{eq:6}) is Lagrangian with respect to $\Omega$. Using the K\"ahler metric induced by the 1-1 form in b), the map (\ref{eq:6}) is locally a map between metric spaces, thus injective. 
\end{proof}

\begin{remark}
In terms of central charge $Z$, the condition a) and b) in proposition 3.1 can be rewritten as
\begin{equation}\label{eq:7}
a)^{\prime} \,\,\,\langle dZ,dZ\rangle=0\,\,\,(\text{transversality})
\end{equation}
\begin{equation}\label{eq:8}
b)^{\prime}\,\,\, \sqrt{-1}\,\langle dZ,d\bar{Z}\rangle>0\,\,\,(\text{nondegeneracy})
\end{equation}
\end{remark}
Choosing a \textit{symplectic basis} $\{\alpha^{i},\beta_{j}\}$ for $\underline{\Gamma}$, and $\{\alpha_{i},\beta^{j}\}$ the dual basis for $\underline{\Gamma}^{\vee}_{b}$,
we write the central charge $Z$ as
\begin{equation}
Z=\sum_{i=1}^{n}a^{i}\alpha_{i}+\sum_{i=1}^{n}a_{D,i}\beta^{i}
\end{equation}
where $a^{i}$, $a_{D,i}\in\mathscr{O}_{\mathcal{B}^{0}}$ are locally defined holomorphic function on $\mathcal{B}^{0}$. The condition (\ref{eq:7}) becomes $d\left(\sum_{i}a_{D,i}\,da^{i}\right)=0$, from which $\sum_{i}a_{D,i}\,da^{i}=d\mathcal{F}$ for some holomorphic function $\mathcal{F}$, called the \textbf{prepotential}. Consequently, We have
\begin{equation}
a_{D,i}=\frac{\partial\mathcal{F}}{\partial a^{i}},\,\,\,i=1,\cdots,n.
\end{equation}
By definition \ref{def:2.4}, we say that the special coordinates $\{a^{i}\}$ and $\{a_{D,i}\}$ form a conjugate pair.

The condition (\ref{eq:8}) becomes: $Im\left(\sum_{i}d{a}_{D,i}\wedge d\bar{a}^{i}\right)>0$, thus the K\"ahler metric on $\mathcal{B}^{0}$ is given as
\begin{equation}\label{eq:11}
    g_{\mathcal{B}^{0}}=Im\left(\sum_{i}d{a}_{D,i}\, d\bar{a}^{i}\right)=Im\left(\sum_{i}d\left(\frac{\partial\mathcal{F}}{\partial a^{i}}\right)\, d\bar{a}^{i}\right)=Im\left(\sum_{i,j}\frac{\partial^{2}\mathcal{F}}{\partial a^{i}\partial a^{j}}\, da^{i}d\bar{a}^{j}\right)
\end{equation}

with the associated \textit{K\"ahler potential} $K=Im\left(\sum_{i}a_{D,i}\,\bar{a}^{i}\right)$. And by definition \ref{def:2.4}, we see that \[x^{i}:=Re(a^{i})\qquad y_{j}:=Re(a_{D,j})\qquad 1\leq i,j\leq n.\] are $\mathbb{Z}$-affine coordinates on the base $\mathcal{B}^{0}$ that are adapted to the special coordinates $a^{i}$ and $a_{D,j}$.

\begin{remark}\label{re3.2}
We can also use $\{Im\, a^{i},Im\,a_{D,j}\}$ as $\mathbb{Z}$-affine coordinates. These define the \textbf{dual affine structure} on $\mathcal{B}^{0}$. It differs from the adapted ones by $\frac{\pi}{2}$-rotation of the central charge $Z$. i.e.,
\begin{equation}\label{eq:12}
    Re\left(e^{-\frac{i\pi}{2}}Z(\gamma_{i})\right)=Im\,Z(\gamma_{i})
\end{equation}
More generally, we consider the rotated central charge $e^{-i\theta}Z_{\gamma}$ by arbitrary angle $\theta\in\mathbb{R}/2\pi\mathbb{Z}$. This would yield a $S^{1}$-family of $\mathbb{Z}$-affine coordinates $\left\{Re(e^{-i\theta}Z(\gamma_{i}))\right\}$ or $\left\{Im(e^{-i\theta}Z(\gamma_{i}))\right\}$.
\end{remark}

\begin{remark}
The rotation of the $\mathbb{Z}$-affine structure is equivalent to rotating the tori $T^{\mathbb{Z}}\mathcal{B}^{0}\otimes\mathbb{R}$, which can be encoded by the cohomology class in $H^{1}(\mathcal{B}^{0},T^{\mathbb{Z}}\mathcal{B}^{0}\otimes i\mathbb{R}/\mathbb{Z}).$ (c.f. \cite{kontsevich2011lectures}).
\end{remark}

\subsection{Monodromy near the discriminant locus}\label{sec:3.3}

Near the discriminant locus $\Delta=\mathcal{B}\backslash\mathcal{B}^{0}$, the Lagrangian fibration $\pi:(X^{0},\omega^{2,0})\to\mathcal{B}^{0}$ degenerates, and some cycles, called the \textit{vanishing cycles}, shrink to zero. This corresponds to the singularities of the $\mathbb{Z}$-affine structure. For the local system $\underline{\Gamma}\to\mathcal{B}^{0}$, the monodromy representation of $\pi_{1}(\mathcal{B}^{0},b)$ into $GL(2n,\mathbb{R})\ltimes\mathbb{R}^{2n}$ is given through
\[\widetilde{\rho}:\pi_{1}(\mathcal{B}^{0},b) \to GL(2n,\mathbb{R})\ltimes\mathbb{R}^{2n},\qquad\gamma\longmapsto\widetilde{\gamma}(1)\]

where $\widetilde{\gamma}$ is the lifting of the loop $\gamma$. The \textit{monodromy group} $M_{b}(\widetilde{\rho})$ based at $b$ is the image of $\widetilde{\rho}$.

\begin{proposition}
(\cite{kontsevich2014wall}) $\widetilde{\rho}$ is equivalent to the data: a) A local system over $\mathcal{B}^{0}$ and\,\,\,b) $\delta\in H^{1}(\mathcal{B}^{0},T^{\mathbb{Z}}\mathcal{B}^{0}\otimes_{\mathbb{Z}}\mathbb{R})$.
\end{proposition}

\begin{remark}
Since $\nabla$ is flat, and $T^{\mathbb{Z}}\mathcal{B}^{0}$ is locally constant, the deRham representation of $\delta$ (still denoted by $\delta$) should be the identity section $id\in \Omega^{1}(\mathcal{B}^{0},T\mathcal{B}^{0})$, i.e. $\delta(v)=v$ for every $v\in T\mathcal{B}^{0}$, more explicitly, we have \[\delta=id=\sum_{i}\frac{\partial}{\partial u^{i}}\otimes du^{i}\]
which is obviously closed by the torsion-freeness of the connection $\nabla$ (see the lemma \ref{lem:2.2}).
\end{remark}

We assume (\textbf{$A^{1}$-singularity assumption}) (section 4.5 of \cite{kontsevich2014wall}) that the singularity $\Delta$ of the $\mathbb{Z}$-affine structure is a two dimensional singularity times $\mathbb{R}^{2n-2}$, and the singularity in dimension two is the simplest possible one---the so called \textit{focus-focus singularity}, i.e., the singular fiber being the \textit{pinched torus}. It is modelled on the \textit{Ooguri-Vafa space}\cite{chan2010ooguri}:

\[\mathcal{B}=\{u\in\mathbb{C}:|u|<\Lambda\}\quad \Delta=\{0\}\]

The charge lattices $\underline{\Gamma}$ is spanned by $\{\gamma_{m}$, $\gamma_{e}\}$, with $\langle\gamma_{e},\gamma_{m}\rangle=1$. The monodromy around the origin is given by 

\[\gamma_{e}(u)\mapsto \gamma_{e}(u),\,\,\,\,
\gamma_{m}(u)\mapsto \gamma_{m}(u)+\gamma_{e}(u).\]

\begin{figure}[h]
    \centering
    \includegraphics[height=3.1cm, width=4.8cm]{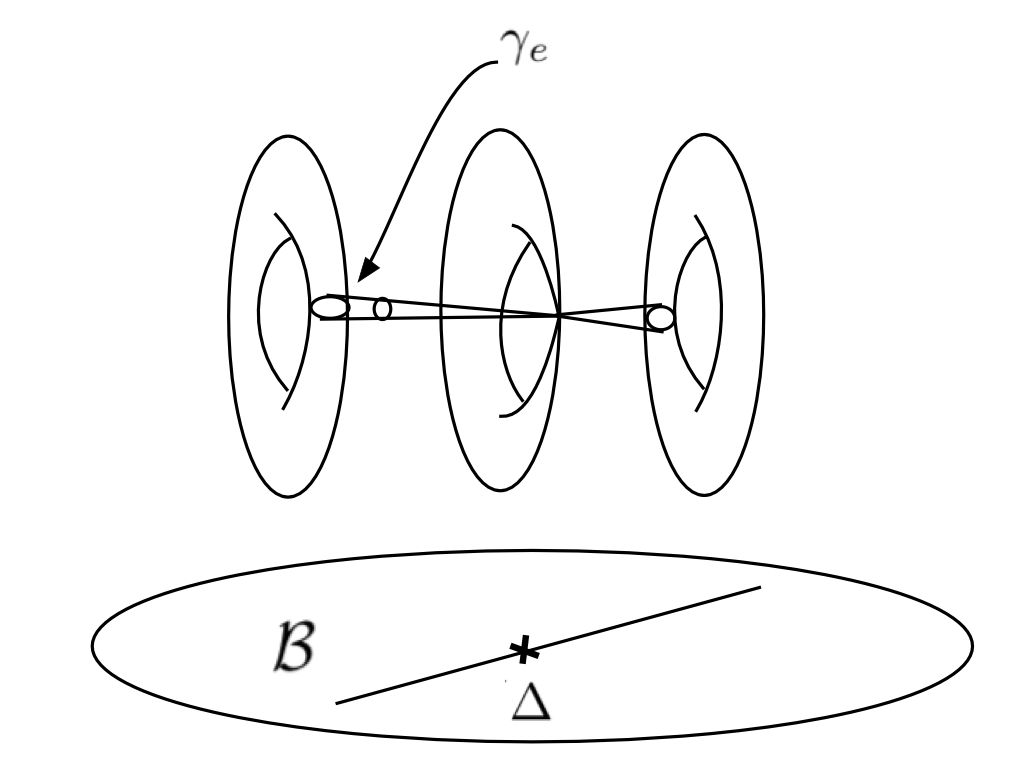}
    \caption{Ooguri-Vafa space}
\end{figure}

\begin{proposition}
The central charge for Ooguri-Vafa space given below satisfies the above monodromy
\begin{equation}\label{eq:13}
Z_{u}(\gamma_{e})=u;\,\,\,\,\,Z_{u}(\gamma_{m})=\frac{1}{2\pi i}\left(u\log\frac{u}{\Lambda}-u\right).
\end{equation}
\end{proposition}

\begin{proof}
From the monodromy above, $Z(\gamma_{e})$ stays the same after looping around the origin, while $Z(\gamma_{m})$ is shifted to $Z(\gamma_{m})+Z(\gamma_{e})$. Indeed, as $u\mapsto e^{2\pi i}u$, one sees easily that\[Z(\gamma_{m})=\frac{1}{2\pi i}\left(u\log\frac{u}{\Lambda}-u\right)\mapsto\frac{1}{2\pi i}\left(e^{2\pi i}u\log\frac{e^{2\pi i}u}{\Lambda}-e^{2\pi i}u\right)=\]
\[=\frac{1}{2\pi i}\left(u\left(2\pi i+\log\frac{u}{\Lambda}\right)-u\right)=u+\frac{1}{2\pi i}\left(u\log\frac{u}{\Lambda}-u\right)=Z(\gamma_{e})+Z(\gamma_{m}).\]
\end{proof}

\subsection{(Split) attractor flow and its properties}

\begin{definition}
Take $\mathbf{u}:=\{a^{i}\}$ as the holomorphic coordinates for $\mathcal{B}^{0}$. Given $\gamma\in\underline{\Gamma}$ near $b_{0}\in\mathcal{B}^{0}$, consider $F_{\gamma}(\mathbf{u}):=Re(Z_{\mathbf{u}}(\gamma))$, the \textbf{attractor flow} associated to $(b_{0},\gamma)\in tot(\underline{\Gamma})$ is given as the gradient flow of $F_{\gamma}(\mathbf{u})$, i.e,
\begin{equation}\label{eq:14}
\dot{\mathbf{u}}+\nabla F_{\gamma}(\mathbf{u})=0
\end{equation}

where $\dot{\mathbf{u}}$ denotes the derivative with respect to the ``time" parameter $t$, and the gradient is taken with respect to the K\"ahler metric (11), i.e., its $i$-th component $\nabla_{i}$ is given by $\sum_{j}g^{i\bar{j}}\,\overline{\partial}_{j}F_{\gamma}(\mathbf{u})$.
\end{definition}

\begin{proposition}\label{pro3.5}
Along the attractor flow, $Im(Z({\gamma}))$ stays constant, thus the flow lines are straight in the $\mathbb{Z}$-affine structure. Away from the discriminant locus $\Delta$, the function $F_{\gamma}(\mathbf{u})$ is decreasing along the flow line.
\end{proposition}

As we noted in remark \ref{re3.2}, $\mathcal{B}^{0}$ is endowed with a $S^{1}$-family of $\mathbb{Z}$-affine structures. Given $\theta\in S^{1}$, the corresponding $\mathbb{Z}$-affine manifold is denoted by $\mathcal{B}^{0}_{\theta}$, and the adapted affine coordinates are given by $Im\left(e^{-i\theta}Z(\gamma)\right)$. Then we see that in the affine structure corresponding to $\theta=Arg Z_{b_{0}}(\gamma)$, the flow line is an affine line on $\mathcal{B}^{0}_{\theta}$ given by 
\begin{equation}\label{eq:15}
Im\left(e^{-i\theta}Z_{\mathbf{u}(t)}(\gamma)\right)=0
\end{equation}

\begin{remark}
We can consider the attractor flow associated to $F_{\gamma}^{\theta}=Re\left(e^{-i\theta}Z_{b}(\gamma)\right)$, where $\theta=Arg(Z_{b_{0}}(\gamma))$, for $(b_{0},\gamma)\in tot\,\underline{\Gamma}$. In this ``rotated case", the attractor line is seen to be given by equation (\ref{eq:15}).
\end{remark}

\begin{definition}\label{def:3.4}
From proposition \ref{pro3.5}, the attractor flow would converge to the local minimum of $F_{\gamma}(\mathbf{u})$, these terminate points of the flow will be called the \textbf{attractor points}.
\end{definition}

Along the flow line, $Re\left(e^{-i\theta}Z_{b}(\gamma)\right)=|Z_{b}(\gamma)|$, thus, the possible minimum of $F_{\gamma}^{\theta}$ are the zeros of $Z(\gamma)$. Since certain homology cycle $\gamma\in\underline{\Gamma}_{b}$ shrinks to zero over $\Delta$, we expect that the attractor points to be located at $\Delta$. By $A^{1}$-singularity assumption, we take $\Delta=\{z_{1}=0\}$. Reducing to two real dimension, denote by $\gamma_{0}$ the vanishing cycle, and by $\gamma_{1}$ the remaining basis element, we see from (\ref{eq:13}) that near the $\Delta$, the central charge take the following form \[Z(\gamma_{0})=z_{1}\,\,\,\,\,\,Z(\gamma_{1})=\frac{1}{2\pi i}\left(z_{1}\log z_{1}+\frac{z_{1}}{2}\right)\]

\begin{proposition}
The attractor flow corresponding to vanishing cycle $\gamma_{0}$ terminates at the singular point where the cycle $\gamma_{0}$ vanishes. For charges other than the vanishing cycles, the attractor flow line avoids the singular point.  
\end{proposition}

\begin{proof}
From the expression of $Z$ above, we see that when we approach to the origin, $Z(\gamma_{0})=z_{1}$ goes to zero, which is the minimum of $Re\left(Z(\gamma_{0})\right)$, thus by proposition \ref{pro3.5} and definition \ref{def:3.4}, the flow line terminates at the attractor point $\Delta$. On the other hand, if the attractor flow associated to $Re\left(Z(\gamma_{1})\right)$ terminates at $\Delta$, then by proposition \ref{pro3.5}, the phase of the central charge would be constant along the flow, but this is impossible due to the form of $Z(\gamma_{1})$ given above. 
\end{proof}

The \textbf{wall of the first kind} $\mathcal{W}_{\gamma}^{1}$ associated to the charge $\gamma$ is given by $\mathcal{W}^{1}_{\gamma}=\bigcup_{\gamma=\gamma_{1}+\gamma_{2}}\mathcal{W}_{\gamma_{1},\gamma_{2}}$, where  \[\mathcal{W}_{\gamma_{1},\gamma_{2}}:=\left\{b\in\mathcal{B}: Im\left(\frac{Z_{b}(\gamma_{1})}{Z_{b}(\gamma_{2}}\right)=0\right\}\]

If the flow line $\mathcal{L}_{\gamma}$ associated to $\gamma$ hits the wall $\mathcal{W}_{\gamma_{1},\gamma_{2}}$ where $Arg (Z(\gamma))=Arg (Z(\gamma_{1}))=Arg (Z(\gamma_{2}))$, then $\mathcal{L}_{\gamma}$ would split into two flow lines $\mathcal{L}_{\gamma_{1}}$ and $\mathcal{L}_{\gamma_{2}}$, associated to $\gamma_{1}$ and $\gamma_{2}$ respectively. 

\begin{figure}[h]
    \centering
    \includegraphics[height=2.9cm, width=6cm]{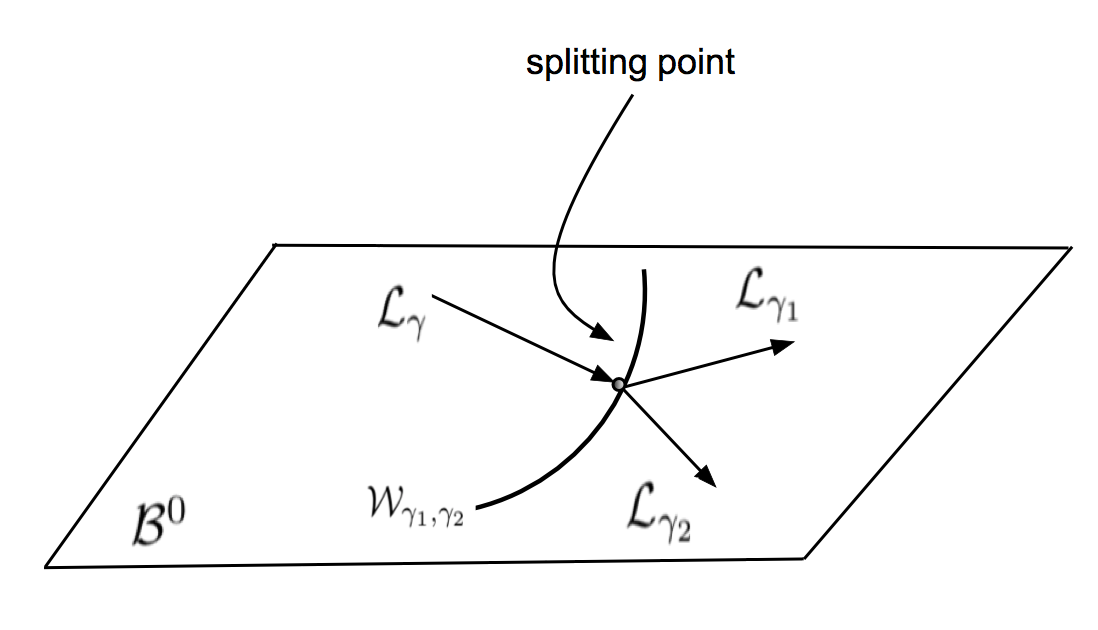}
    \caption{Split attractor flow and splitting point}
    \label{fig:my_label}
\end{figure}

This process could be iterated until the resulting flows terminate at the attractor points. This would give a tree on the base, which we is called the \textbf{(split) attractor tree}.

\subsection{Connection to wall-crossing structure ("WCS" for short)} 

We give a coarse review of the WCS, see \cite{kontsevich2014wall}\cite{wang2020wall} for more about this. It is a formalism used to encode the Donaldson-Thomas invariants (DT invariants) (\cite{kontsevich2008stability}), corresponding to the enumeration of BPS states (c.f., \cite{denef2011split}\cite{wotschke2013bps}\cite{moore2012lecture}) in physics. 

Let $\underline{\Gamma}$ be a local system of charge lattices, and $\underline{\mathfrak{g}}=\bigoplus_{\gamma\in\underline{\Gamma}}\mathfrak{g}_{\gamma}$ a local system of $\underline{\Gamma}$-graded Lie algebras over $\mathcal{B}^{0}$. Given a locally continuous map $Y:\mathcal{B}^{0}\to\Gamma_{\mathbb{R}}^{*}$, then a WCS gives a map $a: tot\,\underline{\Gamma}\to\underline{\mathfrak{g}}$ s.t., $a_{b}(\gamma):=a(b,\gamma)\in\underline{\mathfrak{g}}_{b,\gamma}$ is non-trivial only if $Y(b)(\gamma)=0$ (this defines \textit{wall of the second kind}). 

The map $a$ is a locally constant, and is discontinuous at those $(b,\gamma)$ s.t., $\gamma=\gamma_{1}+\gamma_{2}$ and $Arg (Z_{b}(\gamma_{1}))=Arg(Z_{b}(\gamma_{2}))$. This discontinuous variety projects to $\mathcal{W}^{1}_{\gamma}$ on $\mathcal{B}^{0}$. When the wall $\mathcal{W}^{1}_{\gamma}$ is crossed, the "jumps" of $a_{b}(\gamma)$ are governed by the \textit{Kontsevich-Soibelman wall-crossing formula} (KSWCF) (see \cite{kontsevich2008stability} for more details).

\textbf{The use of split attractor flows}

In \cite{kontsevich2014wall}, it is proposed that we can use the attractor flows to produce an algorithm for computing $a_{b}(\gamma)\in\underline{\mathfrak{g}}_{b,\gamma}$ that satisfy KSWCF. To this end, consider all attractor trees rooted at $(b,\gamma)$ with tail edges hitting the discriminant locus $\Delta$. We expect the number of such trees to be finite. We start from the tail vertices, and move toward to the root $b$, and at each internal vertex $b_{*}$ lying on the first kind wall, we apply KSWCF. By assigning the \textbf{``initial data"} of WCS, i.e., the DT invariants at $\Delta$, we can compute $a_{b}(\gamma)$ inductively by this procedure.

\begin{remark}
From the $A^{1}$-singularity assumption in section \ref{sec:3.3}, we expect that for $b\in\Delta$, there  is only one special direction corresponding the the vanishing cycle over $b$. Consequently, the DT invariants at $b$ is non-trivial only in this special direction, which means that the local system $\underline{\mathfrak{g}}_{loc}$ is typically trivial of rank one. The initial data is given by assigning for $b\in\Delta$ and the primitive vanishing cycle $\gamma$ the group element \cite{kontsevich2014wall}: \[a_{b}(k\gamma):=\frac{1}{k^{2}}e_{k\gamma}\,\,\,\,\,\text{for all}\,\, k\geq 1.\]

This is motivated by the study of multi-cover formula for the Ooguri-Vafa space in \cite{lin2017open}.
\end{remark}

\begin{remark}
In SYZ mirror symmetry program, the WCS is used to glue the local models of the SYZ near the discriminant locus. Roughly speaking, the attractor flows (i.e., tropical curves) and their combinatorial structures on the base $\mathcal{B}$ govern how the holomorphic disks bounded by torus fibers are being glued. The rules of gluing are then determined by certain transformations that enter into the KSWCF associated to the walls. This idea was originally proposed and investigated by Kontsevich and Soibelman in \cite{kontsevich2006affine} in studying the mirror symmetry for $K3$ surfaces.
\end{remark}

\begin{figure}[h]
    \centering
    \includegraphics[height=3.4cm, width=5.6cm]{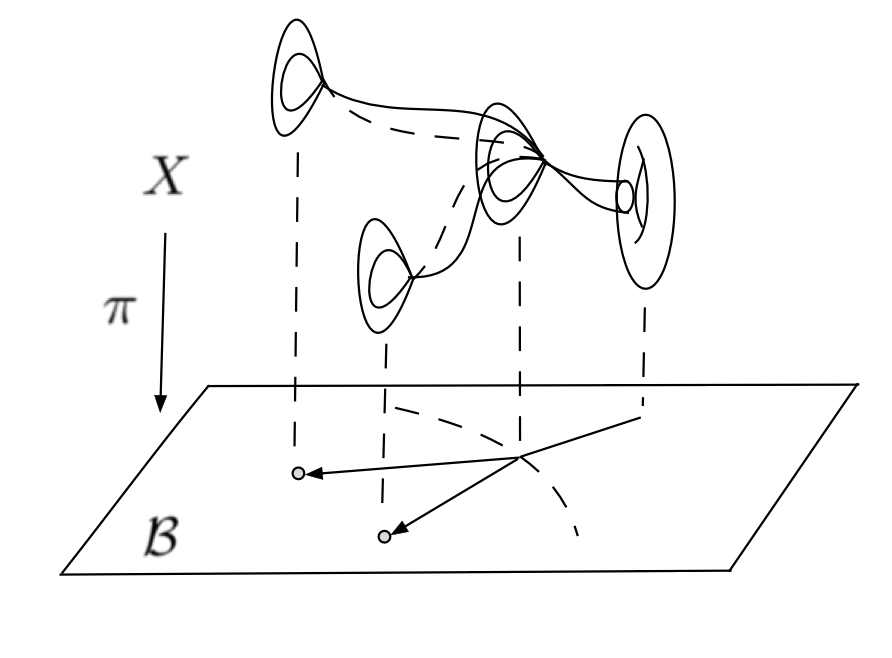}
    \caption{Holomorphic disks and attractor flows}
\end{figure}

\section{Attractor flow versus Hesse flow}\label{sec:4}

\subsection{Legendre transform and the dual Monge-Amp\`er manifold}

We pointed out in section \ref{sec:3} that the base $\mathcal{B}$ of a complex integrable system is naturally endowed with a $\mathbb{Z}$-affine structure with singularities. On the smooth part $\mathcal{B}^{0}$, it is endowed with a K\"ahler metric $g_{\mathcal{B}^{0}}$ (see (\ref{eq:11}))
with K\"ahler potential given by $K=Im\left(\sum_{i}a_{D,i}\,\bar{a}^{i}\right)$. Next, recall (definition \ref{def:2.2}) that a \textit{Monge-Amp\`er manifold} (c.f.,\cite{kontsevich2001homological}) is a triple $(X,g,\nabla)$, where $(X,g)$ is a smooth Riemanian manifold with metric $g$, and $\nabla$ a flat connection on the tangent bundle $TX$ such that $\nabla$ defines an affine structure on $X$, and the metric $g$ in local affine coordinates $(x_{1},\cdots,x_{n})$ can be expressed as \[g_{ij}=\frac{\partial^{2}K}{\partial x_{i}\partial x_{j}}\] for smooth real-valued function $K$, called \textit{Hessian potential}. And $g$ satisfies the real Monge-Amp\`ere equation.

\begin{definition}
Let $U\subset\mathbb{R}^{n}$ be a convex open domain in $\mathbb{R}^{n}$ equipped with the standard affine coordinates $x_{1},\cdots,x_{n}$, and $K:U\to\mathbb{R}$ a convex functions. Then the Legendre transform of the function $K$ is defined as
\begin{equation}
    \widehat{K}(y_{1},\cdots,y_{n}):=\max_{x\in U}\left(\sum_{i}x_{i}y_{i}-K(x_{1},\cdots,x_{n})\right)
\end{equation}
\end{definition}

We have the following proposition for Monge-Amp\`ere manifold (see c.f. lemma 1 in \cite{kontsevich2001homological} for a short elegant proof).

\begin{proposition}
If $K:U\to\mathbb{R}$ is a convex functions satisfying the Monge-Amp\`ere equation,then its Legendre transform $\widehat{K}$ also satisfies the Monge-Amp\`ere equation.
\end{proposition}

From the above proposition, it follows that in the dual affine coordinates $(y_{1},\cdots,y_{n})$, we have another metric $\widehat{g}_{ij}$ (dual to $g_{ij}$) given by \[\widehat{g}_{ij}=\frac{\partial^{2}\widehat{K}}{\partial y_{i}\partial y_{j}}\] We state the following proposition given in \cite{kontsevich2001homological}.

\begin{proposition}
For a given Monge-Amp\`ere manifold $(X,g,\nabla)$ there is a canonically defined dual Monge-Amp\`ere manifold $(\widehat{X},\widehat{g},\widehat{\nabla})$ such that $(X,g)$ is identified with $(\widehat{X},\widehat{g})$ as Riemann manifolds, and the local system $(T_{\widehat{X}},\widehat{\nabla})$ is naturally isomorphic to the local system dual to $(T_{X},\nabla)$. And if $\nabla$ defines an $\mathbb{Z}$-affine structure on $X$, then $\widehat{\nabla}$ defines an $\mathbb{Z}$-affine structure on $\widehat{X}$.
\end{proposition}

Next, we give explicitly the dual Monge-Amp\`ere manifold for the base of the complex integrable system. The computation in the proposition \ref{pro:4.3} below is based on the paper \cite{svan2012bps} of Dieter Van der Bleeken. 

Recall that on the base $\mathcal{B}^{0}$, we have the (rotated) special coordinates $\{e^{-i\theta}a^{k}\}$ and the dual coordinates  $\{e^{-i\theta}a_{D,l}\}$. Consider the adapted real special affine coordinates, i.e.,
\begin{equation}
   e^{-i\theta}a^{k}=x^{k}+iy^{k}, \qquad e^{-i\theta}a_{D,l}=x_{l}+iy_{l} 
\end{equation}

We want to find the relations between these affine coordinates. As $a_{D,l}=\partial\mathcal{F}/\partial a^{l}$, and the prepotential $\mathcal{F}$ is holomorphic in $a^{k}$, we see that $a_{D,l}$ can be viewed as a function in variables $\{x^{k},y^{k}\}$. We claim that 

\begin{proposition}\label{pro:4.3}
The affine variables $x_{l}$ is conjugate to $y^{l}$, while $y_{l}$ is conjugate to $x^{l}$ in the sense that there exists some holomorphic function $\mathscr{F}$ such that
\begin{equation}\label{eq:18}
          x_{l}=\frac{\partial\mathscr{F}}{\partial y^{l}},\qquad y_{l}=\frac{\partial\mathscr{F}}{\partial x^{l}}
\end{equation}
\end{proposition}

\begin{proof} We compute as follows:
\[x_{l}+iy_{l}=e^{-i\theta}a_{D,l}=e^{-i\theta}\frac{\partial\mathcal{F}}{\partial a^{l}}=e^{-i\theta}\left(\sum_{k}\frac{\partial\mathcal{F}}{\partial x^{k}}\frac{\partial x^{k}}{\partial a^{l}}+\sum_{k}\frac{\partial\mathcal{F}}{\partial y^{k}}\frac{\partial y^{k}}{\partial a^{l}}\right)=
\]

\[=e^{-i\theta}\left(e^{-i\theta}\frac{\partial\mathcal{F}}{\partial x^{l}}-i\,e^{-i\theta}\frac{\partial\mathcal{F}}{\partial y^{l}}\right)=\frac{\partial}{\partial x^{l}}\left(e^{-2i\theta}\mathcal{F}\right)-i\,\frac{\partial}{\partial y^{l}}\left(e^{-2i\theta}\mathcal{F}\right)=\]

\[=\frac{\partial}{\partial x^{l}}\left(Re\left(e^{-2i\theta}\mathcal{F}\right)+i\,Im\left(e^{-2i\theta}\mathcal{F}\right)\right)-i\,\frac{\partial}{\partial y^{l}}\left(Re\left(e^{-2i\theta}\mathcal{F}\right)+i\,Im\left(e^{-2i\theta}\mathcal{F}\right)\right)\]

which by applying the Cauchy-Riemann equation for the holomorphic function $e^{-2i\theta}\mathcal{F}$, becomes

\[2\left(\frac{\partial Im\left(e^{-2i\theta}\mathcal{F}\right)}{\partial y^{l}}+i\,\frac{\partial Im\left(e^{-2i\theta}\mathcal{F}\right)}{\partial x^{l}}\right)
\]
By defining $\mathscr{F}:=2Im\left(e^{-2i\theta}\mathcal{F}\right)$, the relations (\ref{eq:18}) then follows from the above computation.
\end{proof}

Then the function $\mathscr{F}$, which is given `apriori' as a function in the variables $(x^{i}, y^{i})$, now can be expressed as a function in the variables $(y_{i},y^{i})$ through the Legendre transform with respect to the conjugate variables $(x^{i},y_{i})$ as follows
\begin{equation}
    \widehat{\mathscr{F}}^{y}(y_{i},y^{i})=\sum_{k}y_{k}x^{k}-\mathscr{F}(x^{i},y^{i})
\end{equation}
It is easy to see that
\begin{equation}\label{eq:20}
    \frac{\partial\widehat{\mathscr{F}}^{y}}{\partial y^{l}}=-x_{l},\qquad \frac{\partial\widehat{\mathscr{F}}^{y}}{\partial y_{k}}=x^{k}
\end{equation}

Thus the special coordinates $\{a^{k}\}$ and its dual $\{a_{D,l}\}$ can be expressed in terms of the affine coordinates $(y^{i},y_{i})$ as
\begin{align}
    a^{k}=a^{k}(y_{i},y^{i})=e^{i\theta}(x^{k}+iy^{k})=e^{i\theta}\left(\frac{\partial\mathscr{\widehat{F}}^{y}}{\partial y_{k}}+iy^{k
    }\right)\\
    a_{D,l}=a_{D,l}(y_{i},y^{i})=e^{i\theta}(x_{l}+iy_{l})=e^{i\theta}\left(-\frac{\partial\widehat{\mathscr{F}}^{y}}{\partial y^{l}}+iy_{l}\right)
\end{align}

Similarly, we can transform $\mathscr{F}$ into a function in variables $(x_{i},x^{i})$ as 
\begin{equation}
    \widehat{\mathscr{F}}^{x}(x_{i},x^{i})=\sum_{k}x_{k}y^{k}-\mathscr{F}(x^{i},y^{i})
\end{equation}
Consequently
\begin{equation}\label{eq:24}
    \frac{\partial\widehat{\mathscr{F}}^{x}}{\partial x^{l}}=-y_{l},\qquad\frac{\partial\widehat{\mathscr{F}}^{x}}{\partial x_{k}}=y^{k}
\end{equation}

Thus, in terms of affine coordinates $(x_{i},x^{i})$, we have the following expressions
\begin{align}
    a^{k}=a^{k}(x_{i},x^{i})=e^{i\theta}(x^{k}+iy^{k})=e^{i\theta}\left(x^{k}+i\frac{\partial\mathscr{\widehat{F}}^{x}}{\partial x_{k}}\right)\\
    a_{D,l}=a_{D,l}(x_{i},x^{i})=e^{i\theta}(x_{l}+iy_{l})=e^{i\theta}\left(x_{l}-i\frac{\partial\widehat{\mathscr{F}}^{x}}{\partial x^{l}}\right)
\end{align}

The function $\widehat{\mathscr{F}}^{x}$ and  $\widehat{\mathscr{F}}^{y}$ above are called the \textbf{Hesse potentials} (see section 4 of \cite{svan2012bps}) as the K\"ahler metric $ds^{2}$ on $\mathcal{B}^{0}$ can be expressed in terms of the Hessian matrices of the two potential functions. Namely, we have
        \[g_{\mathcal{B}^{0}}=\sum_{ij}\,g_{ij}da^{i}\,d\bar{a}^{j}=\sum_{k}\,Im\left(da_{D,k}\,d\bar{a}^{k}\right) =\sum_{k}\,Im\left(e^{i\theta}(dx_{k}+idy_{k})\cdot e^{-i\theta}(dx^{k}-idy^{k})\right)\]
        \begin{equation}
          =\sum_{k}\,dx^{k}\otimes dy_{k}-dx_{k}\otimes dy^{k}  
        \end{equation}
    
In the $\mathbb{Z}$-affine structure on $\mathcal{B}^{0}$ given by the affine coordinates $(x_{i},x^{i})$, we can express the K\"ahler metric (27) as 
\[ g_{\mathcal{B}^{0}}=\sum_{k}\,dx^{k}\otimes dy_{k}-dx_{k}\otimes dy^{k}=\]
\begin{equation*}
        \xlongequal{by\,(\ref{eq:24})}-\sum_{k}\,dx^{k}\otimes\sum_{l}\,\frac{\partial^{2}\widehat{\mathscr{F}}^{x}}{\partial x^{k}\partial x_{l}}dx_{l}-\sum_{k}\,dx_{k}\otimes\sum_{l}\,\frac{\partial^{2}\widehat{\mathscr{F}}^{x}}{\partial x^{k}\partial x_{l}}dx_{l}=
\end{equation*}
\begin{equation}
    =-2\sum_{k,l}\,\frac{\partial^{2}\widehat{\mathscr{F}}^{x}}{\partial x^{k}\partial x_{l}}dx_{k}\otimes dx^{l}
\end{equation}

On the other hand, by utilizing another Hesse potential $\widehat{\mathscr{F}}^{x}$, we have the following  
\[ g_{\mathcal{B}^{0}}=\sum_{k}\,dx^{k}\otimes dy_{k}-dx_{k}\otimes dy^{k}=\]
\begin{equation*}
        \xlongequal{by\,(\ref{eq:20})}\sum_{k}\,\sum_{l}\,\frac{\partial^{2}\widehat{\mathscr{F}}^{y}}{\partial y_{k}\partial y_{l}}dy_{l}\otimes dy_{k}+\sum_{k}\sum_{l}\,\frac{\partial^{2}\widehat{\mathscr{F}}^{y}}{\partial y^{k}\partial y_{l}}dy_{l}\otimes dy^{k}=
\end{equation*}
\begin{equation}
    =2\sum_{k,l}\,\frac{\partial^{2}\widehat{\mathscr{F}}^{y}}{\partial y^{k}\partial y_{l}}dy_{k}\otimes dy^{l}
\end{equation}

Denote by $\mathcal{B}^{0}_{x}$ the manifold $\mathcal{B}^{0}$ in the $\mathbb{Z}$-affine structure with coordinates $(x_{i},x^{i})$, and by $\mathcal{B}^{0}_{y}$ that in the $\mathbb{Z}$-affine structure with coordinates $(y_{i},y^{i})$. The above computation then says that $\mathcal{B}^{0}_{x}$ is endowed with Riemann metric $g_{\mathcal{B}^{0}}^{x}$ given by the Hessian of $-2\widehat{\mathscr{F}^{x}}$, while the metric $g_{\mathcal{B}^{0}}^{y}$ on $\mathcal{B}^{0}_{y}$ is given through the Hessian of $2\widehat{\mathscr{F}^{y}}$. We want to show that the metric $g_{\mathcal{B}^{0}}^{x}$ is dual to $g_{\mathcal{B}^{0}}^{y}$ in the sense that their potentials are Legendre dual to each other.

\begin{proposition}\label{pro:4.4}
The two potentials, namely,
$-2\widehat{\mathscr{F}^{x}}$ and $2\widehat{\mathscr{F}^{y}}$, are the Legendre transform of each other.
\end{proposition}
\begin{proof}
\[\widehat{\widehat{\mathscr{F}^{x}}}(y_{i},y^{i})=\sum_{k}\,\left(x^{k}y_{k}+x_{k}y^{k}\right)-\widehat{\mathscr{F}^{x}}(x_{i},x^{i})=\]\[=\sum_{k}\,\left(x^{k}y_{k}+x_{k}y^{k}\right)-\sum_{k}\,x_{k}y^{k}-\mathscr{F}(x^{i},y^{i})=\]
\[=\sum_{k}\,x^{k}y_{k}+\mathscr{F}(x^{i},y^{i})=\widehat{-\mathscr{F}^{y}}(y_{i},y^{i})\]
Similarly, we have that\[\widehat{\widehat{\mathscr{F}^{y}}}(x_{i},x^{i})=\sum_{k}\,\left(x^{k}y_{k}+x_{k}y^{k}\right)-\widehat{\mathscr{F}^{y}}(x_{i},x^{i})=\]\[
        =\sum_{k}\,\left(x^{k}y_{k}+x_{k}y^{k}\right)-\sum_{k}\,x^{k}y_{k}-\mathscr{F}(x^{i},y^{i})=\]\[=\sum_{k}\,x_{k}y^{k}+\mathscr{F}(x^{i},y^{i})=\widehat{-\mathscr{F}^{x}}(x_{i},x^{i})\]
\end{proof}
Consequently, we see that the affine structures induced by $(x^{i},x_{i})$ and $(y^{i},y_{i})$, together with the corresponding dual metrics $g_{\mathcal{B}^{0}}^{x}$ and $g_{\mathcal{B}^{0}}^{y}$, make the base $\mathcal{B}^{0}$ a Monge-Amp\`ere manifold in two ways which are dual to each other.
\subsection{Attractor flow as Hesse flow}\label{sec:4.2}

Recall that for the charge $\gamma$ at $b_{0}\in\mathcal{B}$ given by $\gamma=(q_{k},g^{k})=\sum_{k}\,q_{k}\alpha^{k}(\mathbf{u})+\sum_{k}\,g^{k}\beta_{k}\left(\mathbf{u}\right)$, it has central charge\[Z_{b_{0}}(\gamma)=\sum_{k}\,q_{k}a^{k}(\mathbf{u})+\sum_{k}\,g^{k}a_{D,k}(\mathbf{u})\]
Denote by $\theta=Arg\,Z_{b_{0}}(\gamma)$, the attractor flow line passing through $b_{0}$, in the direction $\gamma$ is given by (see formulae (\ref{eq:15}))
\[Im\left(e^{-i\theta}Z_{\mathbf{u}(t)}(\gamma)\right)=Im\left(\sum_{k}\,q_{k}\,\left(e^{-i\theta}a^{k}\right)+\sum_{k}\,g^{k}\,\left(e^{-i\theta}a_{D,k}\right)\right)=0\] 
By using the affine coordinates $(y^{i},y_{i})$, the attractor flow equation can be rewritten as
\begin{equation}
    \sum_{k}\,q_{k}\,y^{k}(t)+\sum_{k}\,g^{k}\,y_{k}(t)=0
\end{equation}
which by using the relation (\ref{eq:24}), is equivalent to the following
\begin{equation}\label{eq:31}
    \sum_{k}\,q_{k}\frac{\partial\widehat{\mathscr{F}}^{x}}{\partial x_{k}}-\sum_{k}\,g^{k}\frac{\partial\widehat{\mathscr{F}}^{x}}{\partial x^{k}}=0
\end{equation}
\begin{definition}
For a function $f(x_{i},x^{i})$ in affine structure given by $(x_{i},x^{i})$, define its \textbf{gradient} to be the following
\begin{equation}
    \widetilde{\nabla}f:=\left(\frac{\partial f}{\partial x_{k}},-\frac{\partial f}{\partial x^{k}}\right)
\end{equation}
\end{definition}
Then the \textbf{attractor flow} equation (\ref{eq:31}) above can be written as
\begin{equation}\label{eq:33}
    \widetilde{\nabla}\widehat{\mathscr{F}^{x}}\cdot\gamma=0
\end{equation}
\begin{proposition}
The attractor flow lines are the flow lines on which the gradient of the Hesse potential $\widehat{\mathscr{F}^{x}}$ vanishes.
\end{proposition}

Consider the dual of the attractor flow equation (\ref{eq:33}) above, namely
\begin{equation}\label{eq:34}
     \widetilde{\nabla}\widehat{\mathscr{F}^{y}}\cdot\gamma=0
\end{equation}

By duality, this equation should be interpreted as certain attractor flow called the \textbf{dual attractor flow}. To this end, let us take the real part of the rotated central charge function $e^{-i\theta}Z_{\mathbf{u}(t)}(\gamma)$, which gives
\[Re\left(e^{-i\theta}Z_{\mathbf{u}(t)}(\gamma)\right)=Re\left(\sum_{k}\,q_{k}\,\left(e^{-i\theta}a^{k}\right)+\sum_{k}\,g^{k}\,\left(e^{-i\theta}a_{D,k}\right)\right)= \sum_{k}\,q_{k}\,x^{k}(t)+\sum_{k}\,g^{k}\,x_{k}(t)\]

By using the relations (21) and the definition 4.2., the right hand side of the above equation can be expressed as
\begin{equation}
    \sum_{k}\,q_{k}\frac{\partial\widehat{\mathscr{F}}^{y}}{\partial y_{k}}-\sum_{k}\,g^{k}\frac{\partial\widehat{\mathscr{F}}^{y}}{\partial y^{k}}=:\widetilde{\nabla}\widehat{\mathscr{F}^{y}}\cdot\gamma
\end{equation}
Consequently, we get the following relation between the central charge and Hesse potential
\begin{equation}\label{eq:36}
    Re\left(e^{-i\theta}\,Z(\gamma)\right)=\widetilde{\nabla}\widehat{\mathscr{F}^{y}}\cdot\gamma
\end{equation}

Equation above defines the so-called \textbf{Hesse flow} in \cite{svan2012bps}. However, as $Re\left(e^{-i\theta}\,Z(\gamma)\right)=|Z(\gamma)|$ along the attractor flow line, it can not vanish identically. To make (\ref{eq:36}) compatible with (\ref{eq:34}), we observe that $Re\left(e^{-i(\theta+\frac{\pi}{2})}\,Z(\gamma)\right)\equiv 0$ (see formula \ref{eq:12}), and this leads to the following

\begin{proposition}
The attractor flow on the $\mathbb{Z}$-affine manifold $\mathcal{B}_{\theta}$ given by $\widetilde{\nabla}\widehat{\mathscr{F}^{x}}\cdot\gamma=0$ corresponds to the dual attractor flow $\widetilde{\nabla}\widehat{\mathscr{F}^{y}}\cdot\gamma=0$ on the dual $\mathbb{Z}$-affine manifold $\mathcal{B}_{\theta+\frac{\pi}{2}}$.
\end{proposition}

Similarly, we introduce the notion of the \textbf{dual Hesse flow} described by the following equation:

\begin{equation}\label{eq:37}
    Im\left(e^{-i\theta}Z(\gamma)\right)=\widetilde{\nabla}\widehat{\mathscr{F}^{x}}\cdot\gamma
\end{equation}

which is seen to be nothing but the attractor flow equation when $\theta=Arg\,Z(\gamma)$. 

Consequently, both the Hesse flow and its dual has very simple origin. Namely, by taking the real and imaginary of the rotated central charge function $e^{-i\theta}Z(\gamma)$ (i.e., the rotated special K\"ahler coordinate on $\mathcal{B}$),  and then write them down by using the adapted $\mathbb{Z}$-affine coordinates on $\mathcal{B}$, we will get the Hesse flow and dual Hesse flow respectively. 

\subsection{Relation between the attractor flow and the Hesse flow}\label{sec:4.3}

Motivated by the above discussion, we can treat both attractor flow and Hesse flow in the same framework as below:

We start with central charge function $Z(\gamma)$ associated to the charge $\gamma$, and use the $\mathbb{Z}$-affine coordinates on $\mathcal{B}$, we see that taking the real part of the central charge function gives us the Hesse flow
\begin{equation}\label{eq:38}
    Re\left(e^{-i\theta}Z(\gamma)\right)=\widetilde{\nabla}\widehat{\mathscr{F}^{y}}\cdot\gamma
\end{equation}

while taking the imaginary part gives us the dual Hesse flow
\begin{equation}\label{eq:39}
    Im\left(e^{-i\theta}Z(\gamma)\right)=\widetilde{\nabla}\widehat{\mathscr{F}^{x}}\cdot\gamma
\end{equation}

Next, denote $\theta=Arg\left(Z(\gamma)\right)$. We see that in the $\mathbb{Z}$-affine structure $\mathcal{B}_{\theta}$, the dual Hesse flow (\ref{eq:39}) specializes into the attractor flow (\ref{eq:33}); while in the $\mathbb{Z}$-affine structure $\mathcal{B}_{\theta+\frac{\pi}{2}}$, the Hesse flow (\ref{eq:38}) becomes the dual attractor flow (\ref{eq:34})).

\vspace{7pt}

\begin{remark}
It is known that the $\mathbb{Z}$-affine structures $\mathcal{B}_{\theta}$ and $\mathcal{B}_{\theta+\frac{\pi}{2}}$ correspond, within the frame of mirror symmetry, to the symplectic and complex affine structures respectively (for example, see \cite{lin2017open}). And the duality above,  which manifests in our context as the operations of taking the imaginary and the real part, heuristically speaking, corresponds to separating the "vertical" and the "horizontal" direction respectively.
\end{remark}

\bibliographystyle{unsrt}
\bibliography{references}

\end{document}